\documentclass[11pt,letterpaper]{article}
\usepackage{macros}

\title{Approximating Max-Cut on Bounded Degree Graphs: Tighter Analysis of the FKL Algorithm}
\author{Jun-Ting Hsieh \thanks{Carnegie Mellon University} \and Pravesh K. Kothari\footnotemark[1]}
\date{\today}

\begin{document}
\maketitle

\begin{abstract}
    In this note, we describe a $\alpha_{GW} + \wt{\Omega}(1/d^2)$-factor approximation algorithm for Max-Cut on weighted graphs of degree $\leq d$. Here, $\alpha_{GW}\approx 0.878$ is the worst-case approximation ratio of the Goemans-Williamson rounding for Max-Cut. This improves on previous results for unweighted graphs by Feige, Karpinski, and Langberg~\cite{FKP02} and Flor{\'e}n~\cite{Flo16}. Our guarantee is obtained by a tighter analysis of the solution obtained by applying a natural local improvement procedure to the Goemans-Williamson rounding of the basic SDP strengthened with triangle inequalities. 
\end{abstract}


\section{Introduction}
In 1994, Goemans and Williamson~\cite{GW95} described a polynomial time algorithm based on rounding the natural semidefinite relaxation for the Max-Cut problem to obtain an approximation ratio of $\alphaGW \approx 0.878$. Assuming the Unique Games Conjecture~\cite{Kho02}, this rounding algorithm is in fact worst-case optimal~\cite{KKMO04}.

When the underlying graph has bounded degree, the complexity landscape of Max-Cut is much less clear. In particular, the Goemans-Williamson (GW) rounding algorithm can be improved by an elementary local-search based post-processing step applied to the cut obtained by GW rounding of the SDP relaxation strengthened by adding triangle inequalities. The first such result was shown by Feige, Karpinski, and Langberg~\cite{FKP02} (FKL) to obtain an approximation ratio of $\alphaGW + \epsilon(d)$ for $\epsilon(d) = \Omega(1/d^4)$ for unweighted degree $d$ graphs. Flor{\'e}n~\cite{Flo16} later improved the FKL analysis to obtain $\epsilon(d)=\Omega(1/d^3)$ in the same setting.  

In this note, we give a tighter analysis of the local search scheme to show that $\epsilon(d) = \wt{\Omega}(1/d^2)$ and further extend the result to weighted instances of Max-2LIN.

\begin{theorem} \label{thm:main-theorem}
    There is a polynomial time algorithm that takes input a (weighted) Max-2LIN instance $\varphi$ on a graph with $n$ vertices and degree $\leq d$ and outputs an assignment that satisfies a number of constraints that is within an $\alphaGW + \Omega\Paren{\frac{1}{d^2 \log d}}$ of the optimum.
\end{theorem}

Our analysis of the algorithm differs from FKL's. As in FKL, the high-level plan is to show that when the ``edgewise" analysis of the Goemans-Williamson rounded solution is ``tight" (otherwise, we already obtain an improvement on the worst-case GW guarantee), the rounded solution is locally suboptimal. That is, there is a $c_d$-fraction of vertices, for some constant $c_d$ depending only on $d$, such that switching them to the other side should cut more of the edges to their neighbors and thus increase the cut by an edge. But flipping a vertex may kill the increase from others. So we may not be able to ``win" from every one of the candidates. The analysis then involves managing this dependency and lower bounding $c_d$. Specifically, FKL and Flor{\'e}n accomplish this by analyzing the chance that more than half of a vertex's neighbors lie on the same side of the partition generated by GW rounding. In such a plan, they have to handle vertices of odd and even degrees slightly differently and generalizations to weighted graphs seem to lose additional factors in the approximation ratio.

Instead of trying to gain only a single edge from a vertex flip, our analysis directly focuses on lower bounding the quantitative gain in the weight of the cut by such an operation. This strategy also allows us to easily generalize our results to Max-2LIN and to arbitrary weighted graphs.

There is still a wide gap between approximation algorithms and hardness for Max-Cut in the bounded degree setting. In particular, for every $\epsilon>0$, the best known hardness result~\cite{Tre01} rules out polynomial time algorithms with an approximation ratio $>\alpha_{GW} + 5/\sqrt{d}+\epsilon$ for graphs of large enough constant degree $d$. Improving the approximation ratio to get closer to this bound (this will likely need a scheme different from FKL's) or obtaining a better inapproximability results are outstanding open questions.
\section{The Algorithm and Analysis}

\begin{mdframed}
    \begin{algorithm}[Max-2LIN on bounded degree graphs]
      \label{alg:max-cut}\mbox{}
      \begin{description}
          \item[Given:]
              A graph $G = (V,E)$ on $n$ vertices, $m$ edges and maximum degree $d$, signs $b_{ij} \in \pmo$, and non-negative weights $w_{ij} > 0$ for each $\{i,j\}\in E$.
          \item[Operation:]\mbox{}
              \begin{enumerate}
                    \item Find unit vectors\footnote{Approximately solving an SDP in $\poly(n)$ time produces such vectors with value at most $2^{-\poly(n)}$ smaller.}  $v_1,\dots, v_n \in \R^n$ that maximize $\frac{1}{m}\sum_{\{i,j\} \in E} w_{ij} \frac{1}{2}(1 + b_{ij} \iprod{v_i,v_j})$, and, satisfy $\Norm{a_i v_i - a_j v_j}^2 + \Norm{a_j v_j - a_k v_k}^2 \geq \Norm{a_i v_i - a_k v_k}^2$ for every triple $\{i,j,k\}$ and $(a_i, a_j, a_k)\in \pmo^3$.
                    \item Sample $g \sim \calN(0, \Id_n)$\footnote{As in standard implementations of the GW rounding scheme, truncating Gaussian samples to rationals of $\poly(n)$-bits suffices to recover the stated guarantees up to a loss of an additive $2^{-\poly(n)}$ in the approximation ratio. We will omit a detailed discussion of issues of numerical precision in this note.} , and set $x_i = \sgn(\iprod{g, v_i}) \in \pmo$.
                    \label{step:hyperplane-rounding}

                    \item Set $\eps = \frac{1}{Cd \sqrt{\log d}}$ for some large enough constant $C>0$.
                    Define $S \coloneqq \{i\in V: \iprod{g, v_i} \in (-\eps,\eps)\}$ to be the \emph{candidate set} of suboptimal vertices.
                    For each $i\in S$, partition its neighbors $N(i)$ into 3 disjoint sets:
                    \begin{enumerate}[(a)]
                        \item $A_i = N(i) \cap S = \{j\in N(i): \iprod{g, v_j} \in (-\eps,\eps)\}$,
                        \item $B_i = \{j\in N(i): b_{ij} x_i \iprod{g, v_j} \leq -\eps\}$,
                        \item $C_i = \{j\in N(i): b_{ij} x_i \iprod{g, v_j} \geq \eps\}$.
                    \end{enumerate}
                    \label{step:partition-neighbors}

                    \item Output $x'$ obtained by flipping $x_i$ for every $i\in S$ such that $\sum_{j\in B_i} w_{ij} > \sum_{j\in A_i \cup C_i} w_{ij}$.
                  
                    \label{step:flip}
              \end{enumerate}
      \end{description}
    \end{algorithm}
\end{mdframed} 

\newpage
\paragraph{Analysis} We will appeal to the following three elementary facts. 

\begin{fact}[Gaussian Measure in the Core]\label{fact:g-in-eps}
    For any $\eps > 0$,
    \begin{equation*}
        \frac{\eps}{\sqrt{2\pi}} e^{-\eps^2/2} \leq \Pr_{g\sim \calN(0,1)}\Brac{g\in(0,\eps)}
        \leq \frac{\eps}{\sqrt{2\pi}} \mper
    \end{equation*}
\end{fact}

\begin{fact}[Sheppard's Lemma] \label{fact:g1-g2-positive}
    Let $g_1,g_2$ be standard Gaussian variables with covariance $\sigma = \E[g_1 g_2] \in [-1,1]$.
    Then, $\Pr[g_1 \geq 0 \wedge g_2 \geq 0] = \frac{1}{2} - \frac{1}{2\pi}\arccos(\sigma)$.
\end{fact}

\begin{fact}[Corollary of Schur Product Theorem] \label{fact:hadamard-psd}
    Let $A \in \R^{n\times n}$ be positive semidefinite.
    For any $t \in \N$, the matrix $B\in \R^{n\times n}$ with entries $B_{ij} = A_{ij}^t$ is positive semidefinite.
\end{fact}

Our analysis crucially relies on the following basic facts about the  arcsin function.
\begin{lemma}[Basic Facts about Taylor Approximation for $\arcsin$] \label{lem:arcsin-analysis}
    Let the Taylor expansion of the $\arcsin$ function be $\arcsin(x) = \sum_{k=0}^{\infty} c_k x^{2k+1}$ for $|x| \leq 1$ where $c_k = \frac{(2k)!}{2^{2k}(k!)^2(2k+1)}$.
    For any $\tau \in \N$, $\tau > \tau_0$ where $\tau_0$ is a universal constant,
    \begin{enumerate}[(1)]
        \item for $x > 0$, $\arcsin(x) \geq \sum_{k=0}^{\tau} c_k x^{2k+1}$,
        \label{item:arcsin-1}

        \item for $|x| \leq 1/2$, $\arcsin(x) \geq \sum_{k=0}^{\tau} c_k x^{2k+1} - O(\tau^{-1/2} \cdot 2^{-2\tau})$,
        \label{item:arcsin-2}

        \item for $x = 1$, $\arcsin(1) = \frac{\pi}{2} = \sum_{k=0}^{\tau} c_k + \Theta(\tau^{-1/2})$.
        \label{item:arcsin-3}
    \end{enumerate}
\end{lemma}
\begin{proof}
    We first show that $c_k = \Theta(k^{-3/2})$ by Stirling's approximation:
    \begin{equation*}
        c_k = (1+O(1/k)) \cdot \frac{\sqrt{4\pi k} (2k/e)^{2k}}{2^{2k} \cdot 2\pi k (k/e)^{2k} \cdot (2k+1)} = \frac{1}{2\sqrt{\pi}} k^{-3/2} (1+O(1/k)) \mper
    \end{equation*}
    Therefore, for any $\tau > \tau_0$ where $\tau_0$ is a universal constant,
    \begin{equation*}
        \sum_{k=\tau}^\infty c_k = \Theta(1) \sum_{k=\tau}^\infty k^{-3/2}
        = \Theta(1) \int_{\tau}^\infty x^{-3/2}\ dx
        = \Theta(\tau^{-1/2}) \mper
    \end{equation*}
    Now we prove the lemma.
    \ref{item:arcsin-1} is straightforward because $c_k > 0$ for all $k$.
    \ref{item:arcsin-2} holds since
    \begin{equation*}
        \Abs{\sum_{k=\tau+1}^\infty c_k x^{2k+1} }
        \leq \sum_{k=\tau+1}^\infty c_k |x|^{2k+1}
        \leq 2^{-2\tau}\sum_{k=\tau+1}^{\infty} c_k
        \leq O(\tau^{-1/2} \cdot 2^{-2\tau}) \mper
    \end{equation*}
    Finally, \ref{item:arcsin-3} follows directly from $\sum_{k=\tau+1}^\infty c_k = \Theta(\tau^{-1/2})$.
\end{proof}

\pref{lem:arcsin-analysis} states that for some large threshold $\tau$, the Taylor approximation error for $|x| \leq 1/2$ is a factor $2^{-2\tau}$ smaller than the error for $x=1$.
This gap allows us to prove the key lemma below:

\begin{lemma} \label{lem:sum-of-arcsin}
    Let $d\in \N$, $d\geq 2$.
    Let $A \in \R^{d\times d}$ be a positive semidefinite matrix such that $A_{ii} = 1$ and $A_{ij} \geq -1/2$ for all $i, j\in [d]$.
    Let $w \in \R^d_{\geq 0}$ be a vector with non-negative entries.
    Then,
    \begin{equation*}
        \sum_{i, j=1}^d w_i w_j \arcsin(A_{ij}) \geq \Omega\Paren{\frac{\|w\|_1^2}{d\sqrt{\log d}}} \mper
    \end{equation*}
\end{lemma}
\begin{proof}
    Pick a threshold $\tau = C \log_2 d$ for a large enough constant $C$.
    Since we have the assumption that $A_{ij} \geq -1/2$ and $w_i \geq 0$ for all $i,j$, we can bound the off-diagonal entries using \ref{item:arcsin-1} and \ref{item:arcsin-2} of \pref{lem:arcsin-analysis},
    \begin{equation*}
    \begin{aligned}
        \sum_{i \neq j}^d w_i w_j \arcsin(A_{ij})
        &\geq \sum_{i\neq j}^d w_i w_j \Paren{ \sum_{k=0}^{\tau} c_{k} A_{ij}^{2k+1} - O(\tau^{-1/2}2^{-2\tau}) } \\
        &= \sum_{k=0}^{\tau} c_{k} \sum_{i \neq j}^d w_i w_j A_{ij}^{2k+1} - \wt{O}\Paren{\frac{\|w\|_1^2}{d^{2C}}} \mper
    \end{aligned}
    \end{equation*}
    The diagonal entries are $\arcsin(1) = \sum_{k=0}^\tau c_k + \Theta(\tau^{-1/2})$ by \ref{item:arcsin-3} of \pref{lem:arcsin-analysis}.
    Thus, we get
    \begin{equation*}
    \begin{aligned}
        \sum_{i, j=1}^d w_i w_j \arcsin(A_{ij})
        &\geq \sum_{k=0}^\tau c_k \sum_{i,j=1}^d w_i w_j A_{ij}^{2k+1} + \sum_{i=1}^d w_i^2 \cdot \Omega(\tau^{-1/2}) - \wt{O}\Paren{\frac{\|w\|_1^2}{d^{2C}}} \\
        &= \sum_{k=0}^\tau c_k \sum_{i,j=1}^d w_i w_j A_{ij}^{2k+1}  + \Omega\Paren{\frac{\|w\|_2^2}{\sqrt{\log d}}} - \wt{O}\Paren{\frac{\|w\|_1^2}{d^{2C}}} \mper
    \end{aligned}
    \end{equation*}
    Since $A \succeq 0$, by \pref{fact:hadamard-psd} we know that $\sum_{ij=1}^d w_i w_j A_{ij}^{2k+1} \geq 0$ for all $k \geq 0$.
    Finally, any $w\in \R^d$ satisfies $\|w\|_2^2 \geq \frac{1}{d}\|w\|_1^2$.
    This completes the proof.
\end{proof}

\subsection{Proof of \pref{thm:main-theorem}}

Let $\rho_* = \argmin_{\rho\in[-1,1]} \frac{1+ \frac{2}{\pi}\arcsin(\rho)}{1+\rho} \approx 0.689$ such that $\frac{1+ \frac{2}{\pi}\arcsin(\rho_*)}{1+\rho_*} = \alphaGW$.
Following \cite{FKP02}, we can assume without loss of generality that for all $(i,j)\in E$ with sign $b_{ij}$, the SDP solution satisfies $b_{ij} \iprod{v_i,v_j} \in [\rho_* - 0.01, \rho_* + 0.01]$.

Observe that after \pref{step:hyperplane-rounding} of \pref{alg:max-cut}, for every $i$ in the candidate set $S$, by definition all edges between $i$ and $B_i$ are violated, while all edges between $i$ and $C_i$ are satisfied.
Moreover, it is important that $B_i$ and $C_i$ are \emph{disjoint} from $S$, so their assignments will not be flipped in \pref{step:flip}.
For edges between $i$ and $A_i$, in the worst case all of them are violated after flipping.
Thus, if $\sum_{j\in B_i} w_{ij} > \sum_{j\in A_i \cup C_i} w_{ij}$, then flipping $x_i$ will increase the Max-2LIN value.


We will prove that the expected gains from such local updates are large.
For a vertex $i\in S$, let $W_i \coloneqq \sum_{j\in N(i)} w_{ij}$ be the total weight of edges incident to $i$.
Define the \emph{local gain} from $i$ to be
\begin{equation*}
    \Delta_i \coloneqq \Paren{\sum_{j\in B_i} w_{ij} - \sum_{j\in A_i \cup C_i} w_{ij}}_+ 
    = \Paren{2\sum_{j\in B_i} w_{ij} - W_i}_+ \mcom
\end{equation*}
where we denote $(z)_+ = \max(0, z)$.
The following is the key lemma of our analysis showing that conditioned on $i\in S$, the expected local gain from vertex $i$ is at least $\wt{\Omega}(\frac{W_i}{d})$.

\begin{lemma}[Expected local gain] \label{lem:local-gain}
    Let $d\in \N, d\geq 2$ and $\eps = \frac{1}{C d \sqrt{\log d}}$ for a large enough constant $C$.
    For a vertex $i\in V$ with degree $d$, the expected local gain $\E[\Delta_i| i\in S] \geq \Omega\Paren{\frac{W_i}{d\sqrt{\log d}}}$.
\end{lemma}

\begin{proof}
    Consider vertex $i\in S$ and its $d$ neighbors, denoted $[d]$.
    We first introduce some notations for the analysis.
    \begin{itemize}
        \item We can assume that $v_i = (1,0,\dots,0)$ without loss of generality due to rotational symmetry.
        
        \item For every neighbor $j\in [d]$, let $v_j = (\rho_j, v_j')$ where the first coordinate is $\rho_j \in b_{ij} \cdot [\rho_* \pm 0.01]$ and $v_j' \in \R^{n-1}$ since we assume that $b_{ij} \iprod{v_i, v_j} \approx \rho_*$ for all $(i,j)\in E$.

        \item Denote the unit vector $\wh{v}_j = v_j'/\|v_j'\|_2$.

        \item Let $g = (g_1,\dots,g_n) \sim \calN(0, \Id_n)$ be the sampled Gaussian vector, and let $g' =(g_2,\dots,g_n)$.
        
        \item Let $h_j \seteq b_{ij} x_i \iprod{\wh{v}_j, g'}$.
        The random vector $h = (h_1,\dots,h_d)$ is a multivariate Gaussian variable with covariance matrix $\Sigma \in \R^{d\times d}$ where $\Sigma_{jj} = 1$ and $\Sigma_{jk} = b_{ij} b_{ik} \iprod{\wh{v}_j, \wh{v}_k}$.

    \end{itemize}

    \noindent
    Since $\|v_j\|_2 = 1$, we have $\|v_j'\|_2 = \sqrt{1-\rho_j^2} \in [0.726 \pm 0.01]$, and $\iprod{v_j, g} = \rho_j g_1 + \sqrt{1-\rho_j^2} \iprod{\wh{v}_j,g'}$.
    Recall that $S = \{i\in V: \iprod{g, v_i} \in (-\eps,\eps)\}$.
    Thus, $i\in S$ means that $\iprod{v_i, g} = g_1 \in (-\eps,\eps)$.
    Next, we define the following random variable
    \begin{equation*}
        Z \coloneqq \sum_{j=1}^d w_{ij} \cdot \1\Paren{h_j \leq -3\eps } \mper
    \end{equation*}
    Conditioned on the event that $|g_1| < \eps$,
    \begin{equation*}
        h_j \leq -3\eps\
        \Longrightarrow\ b_{ij} x_i \iprod{v_j, g} \leq |\rho_j g_1| + \sqrt{1-\rho_j^2} \cdot h_j \leq -\eps \mper
    \end{equation*}
    Therefore, we have $Z \leq \sum_{j\in B_i} w_{ij}$
    (recall that $B_i = \{j\in N(i): b_{ij} x_i \iprod{v_j,g} \leq -\eps\}$).
    Then,
    \begin{equation*}
        \Delta_i = \Paren{2\sum_{j\in B_i} w_{ij} - W_i}_+
        \geq (2Z- W_i)_+ \mper
    \end{equation*}
    Thus, it suffices to lower bound $\E[(2Z-W_i)_+]$.

    First, every $h_j$ is a standard Gaussian, so let $p \seteq \Pr[h_j \leq -3\eps] = \frac{1}{2}- c\eps$ for some $c \leq \frac{3}{\sqrt{2\pi}}$ by \pref{fact:g-in-eps}.
    Then, $\E[Z] = pW_i$.
    Next, we lower bound $\E[(Z-W_i/2)^2]$.

    \begin{equation*} \label{eq:Z-W2-squared}
    \begin{aligned}
        \E\Brac{(Z-W_i/2)^2} &= \E \Paren{\sum_{j=1}^d w_{ij} \Paren{\1(h_j\leq -3\eps) - \frac{1}{2}} }^2 \\
        &= \frac{1}{4} \sum_{j=1}^d w_{ij}^2 + \sum_{j \neq k}^d w_{ij} w_{ik} \Paren{ \Pr\Brac{h_j \leq -3\eps \wedge h_k \leq -3\eps} - p + \frac{1}{4} } \\
        &\geq \frac{1}{4} \sum_{j=1}^d w_{ij}^2 + \sum_{j\neq k}^d w_{ij} w_{ik} \Paren{\frac{1}{2} - \frac{1}{2\pi} \arccos(\Sigma_{jk}) - 2c\eps - \Paren{\frac{1}{4} - c\eps} } \\
        &\geq \frac{1}{2\pi} \sum_{j,k = 1}^d w_{ij} w_{ik} \arcsin(\Sigma_{jk}) - c\eps W_i^2 \mper
    \end{aligned}
    \numberthis
    \end{equation*}
    The third line follows from $\Pr[h_j \leq -3\eps \wedge h_k \leq -3\eps] \geq \Pr\Brac{h_j \leq 0 \wedge h_k \leq 0} - 2\cdot \Pr[-3\eps \leq h_j  \leq 0]$ and applying \pref{fact:g1-g2-positive}.
    The final inequality is because $\frac{\pi}{2}- \arccos(\theta) = \arcsin(\theta)$ and $\arcsin(\Sigma_{jj}) = \arcsin(1) = \frac{\pi}{2}$.

    By the triangle inequality of the SDP solution, for any $j \neq k$,
    $\|v_i - b_{ij} v_j\|^2 + \|v_i - b_{ik} v_k\|^2 \geq \|b_{ij} v_j - b_{ik} v_k\|^2$.
    Expanding this, we get
    \begin{equation*}
        b_{ij} b_{ik} \Paren{\rho_j \rho_k + \sqrt{1-\rho_j^2}\sqrt{1-\rho_k^2} \iprod{\wh{v}_j, \wh{v}_k}} \geq b_{ij} \rho_j + b_{ik} \rho_k - 1 \mper
    \end{equation*}
    Since $b_{ij} \rho_j \in [\rho_* \pm 0.01]$ for all $j\in[d]$, we have
    \begin{equation*}
        \Sigma_{jk} = b_{ij} b_{ik} \iprod{\wh{v}_j, \wh{v}_k} \geq -0.2 \mper
    \end{equation*}
    This is crucial since we can now apply \pref{lem:sum-of-arcsin} to \pref{eq:Z-W2-squared} and get
    \begin{equation*}
        \E\Brac{(Z-W_i/2)^2} \geq \Omega\Paren{\frac{W_i^2}{d\sqrt{\log d}}} - O(\eps W_i^2) \mper
    \end{equation*}

    Finally, we lower bound $\E[(2Z - W_i)_+]$.
    Let $\ol{Z} = Z - W_i/2$, and let $\ol{Z}_+ = \max(0, \ol{Z})$ and $\ol{Z}_- = \max(0, -\ol{Z})$.
    Thus, $\ol{Z} = \ol{Z}_+ - \ol{Z}_-$ and $\ol{Z}^2 = \ol{Z}_+^2 + \ol{Z}_-^2$ by definition, and both $\ol{Z}_+$ and $\ol{Z}_-$ lie in $[0, W_i/2]$.
    Furthermore, $\E[\ol{Z}] = \E[\ol{Z}_+] - \E[\ol{Z}_-] = pW_i - W_i/2 = -c \eps W_i$.
    Then,
    \begin{equation*}
    \begin{aligned}
        \E \Brac{\ol{Z}^2} &= \E\Brac{\ol{Z}_+^2 + \ol{Z}_-^2} 
        \leq \frac{W_i}{2} \cdot \E\Brac{\ol{Z}_+ + \ol{Z}_-}
        = W_i \cdot \E[\ol{Z}_+] + \frac{c}{2} \eps W_i^2 \mper
    \end{aligned}
    \end{equation*}
    Setting $\eps \leq \frac{1}{C d\sqrt{\log d}}$ for a large enough $C$, we have $\E[\ol{Z}_+] \geq \Omega\Paren{\frac{W_i}{d\sqrt{\log d}}}$.
    This completes the proof.
\end{proof}

We can now prove \pref{thm:main-theorem}.

\begin{proof}[Proof of \pref{thm:main-theorem}]
    We first assume that for all $(i,j)\in E$ with sign $b_{ij}$, the SDP solution satisfies $b_{ij} \iprod{v_i,v_j} \in [\rho_* - 0.01, \rho_* + 0.01]$ where $\rho_* \approx 0.689$.
    Recall that in \pref{alg:max-cut}, for every $i\in S$,
    all edges between $i$ and $B_i$ are violated, and all edges between $i$ and $C_i$ are satisfied.
    Further, since $B_i$ and $C_i$ are disjoint from $S$, the vertices in $B_i$ and $C_i$ will not be flipped.
    For edges between $i$ and $A_i$, in the worst case all of them are violated after flipping.
    Thus, we have a local gain of at least $\Delta_i = (\sum_{j\in B_i} w_{ij} - \sum_{j\in A_i \cup C_i} w_{ij})_+$.

    The expected total gain from the local updates (over the random sample of $g$) is
    \begin{equation*}
        \E[\Delta] = \E \sum_{i \in S} \Delta_i = \sum_{i\in V} \E \Brac{\1(i\in S) \Delta_i} = \sum_{i\in V} \Pr[i\in S] \cdot \E[\Delta_i | i\in S] \geq \Omega(\eps) \sum_{i\in V} \E[\Delta_i | i\in S] \mcom
    \end{equation*}
    where $\Pr[i\in S] = \Omega(\eps)$ is due to \pref{fact:g-in-eps}.
    Then, setting $\eps = \frac{1}{C d \sqrt{\log d}}$ for a large enough constant $C$, by \pref{lem:local-gain} and the fact that the total weight $W = \frac{1}{2}\sum_{i\in V} W_i$, we have
    \begin{equation*}
        \E[\Delta] \geq \Omega(\eps) \cdot \sum_{i\in V}  \Omega\Paren{\frac{W_i}{d\sqrt{\log d}}} = W \cdot \Omega\Paren{\frac{1}{d^2 \log d}} \mper
    \end{equation*}
    Therefore, if the Max-2LIN instance $\varphi$ has optimum $\OPT \leq W$, then in expectation we can find an assignment that satisfies
    \begin{equation*}
        \varphi(x) \geq \alphaGW \cdot \OPT + W\cdot \Omega\Paren{\frac{1}{d^2 \log d}}
        \geq \Paren{\alphaGW + \Omega\Paren{\frac{1}{d^2 \log d}}} \cdot \OPT \mper
    \end{equation*}
    
    From here on, we follow the same argument as \cite{FKP02}.
    Let $\delta = \Omega(\frac{1}{d^2 \log d})$ be the improvement above.
    If more than $\delta/2$ fraction of the (weighted) edges satisfy $b_{ij} \iprod{v_i, v_j} \notin [\rho_* \pm 0.01]$, then hyperplane rounding already gives us $\alphaGW + \Omega(\delta)$ approximation ratio.
    If at most $\delta/2$ fraction of the edges have $b_{ij}\iprod{v_i, v_j} \notin [\rho_* \pm 0.01]$, then we can simply ignore those edges and get an approximation of $(1-\delta/2)(\alphaGW + \delta) \geq \alphaGW + \Omega(\delta)$.
    This completes the proof.
\end{proof}

\section*{Acknowledgement}
We thank Prasad Raghavendra for his  Simons Institute talk\footnote{\url{https://simons.berkeley.edu/talks/title-tba-5}} on approximation algorithms for bounded degree constraint satisfaction problems and various related discussions that directly motivated this work. We thank Simons Institute Berkeley for hosting us in the Fall 2021 research program on Computational Complexity of Statistical Inference.
\bibliographystyle{alpha}
\bibliography{main}

\end{document}